\documentclass[a4paper,UKenglish]{lipics-v2016}
\usepackage{microtype}
\usepackage{amsmath}
\usepackage{algorithmicx}
\usepackage{algorithm}
\usepackage{algpseudocode}

\algnewcommand{\LineComment}[1]{\State \(\triangleright\) #1}
\algnewcommand\algorithmicinput{\textbf{Input:}}
\algnewcommand\INPUT{\item[\algorithmicinput]}
\algnewcommand\algorithmicoutput{\textbf{Output:}}
\algnewcommand\OUTPUT{\item[\algorithmicoutput]}

\DeclareMathOperator*{\argmax}{arg\,max}
\DeclareMathOperator*{\argmin}{arg\,min}

\newcommand{\ce}[1]{\lceil #1 \rceil}
\newcommand{\opt}{\text{OPT}}
\newcommand{\vs}{\text{VS}}

\title{
An FPTAS of Minimizing Total Weighted Completion Time on Single Machine with Position Constraint
}
\titlerunning{Single Machine Scheduling with Position Constraint}

\author[1]{Gruia C\v{a}linescu}
\author[2]{Florian Jaehn}
\author[3]{Minming Li}
\author[4]{Kai Wang}
\affil[1]{Department of Computer Science, Illinois Institute of Technology, Chicago, IL 60616, USA\\
\texttt{calinescu@iit.edu}}
\affil[2]{Management Science and Operations Research, Helmut Schmidt University -- University of the Federal Armed Forces Hamburg, Holstenhofweg 85, D-22043 Hamburg, Germany\\
\texttt{florian.jaehn@hsu-hh.de}
}
\affil[3]{Department of Computer Science, City University of Hong Kong, 83 Tat Chee Avenue, Kowloon, Hong Kong SAR, China\\
\texttt{minming.li@cityu.edu.hk}
}
\affil[4]{Department of Computer Science, City University of Hong Kong, 83 Tat Chee Avenue, Kowloon, Hong Kong SAR, China\\ \texttt{kai.wang@my.cityu.edu.hk}
}

\authorrunning{G. C\v{a}linescu et. al.} 

\Copyright{Gruia C\v{a}linescu, Florian Jaehn, Minming Li and Kai Wang}



\EventEditors{Yoshio Okamoto and Takeshi Tokuyama} \EventNoEds{2}
\EventLongTitle{28th International Symposium on Algorithms and Computation (ISAAC 2017)} \EventShortTitle{ISAAC 2017}
\EventAcronym{ISAAC} 
\EventYear{2017} 
\EventDate{December 9--12, 2017} \EventLocation{Phuket, Thailand} 
\EventLogo{} 
\SeriesVolume{92} 
\ArticleNo{19}

\begin{document}
\maketitle

\begin{abstract}
In this paper we study the classical scheduling problem of minimizing the total weighted completion time on a single machine 
with the constraint that one specific job must be scheduled at a specified position.
We give dynamic programs with pseudo-polynomial running time, and a fully polynomial-time approximation scheme (FPTAS).
\end{abstract}

\section{Introduction}
In general, a major challenge of scheduling problems is the determination of a job sequence for each machine involved. Especially in non-preemptive one machine settings without idle times, this is usually the only task to be performed. In this context, scheduling problems appear without restrictions on this sequence (e.g.\ $1||\sum T_j$, minimize total tardiness) or with restrictions on the sequence (e.g.\ $1|r_j, prec|\sum C_j$, minimize total completion time). Restrictions on the sequence are commonly either time dependent or linked to job pairs. Examples for time dependent restrictions are release dates, deadlines, or time dependent maintenance activities. Precedence constraints are a typical restriction based on job pairs. In this paper, we elaborate on a different restriction based on the position of a job in the sequence. To be more precise, we force one job to have a fixed position within the sequence of jobs. 

The practical and theoretical motivation for such a scheduling problem is twofold. Firstly, such a job that has a fixed position in the sequence could be considered as a maintenance operation. 
Classically, maintenance is also considered to be time dependent, e.g.\ by modeling predetermined machine unavailability intervals (\cite{lee1996, blazewicz2003, low2010}), by allowing a maximum time between two maintenances, which is often referred to as ``tool changes'' (\cite{chen2008, costa2016}), or maintenances may be inserted arbitrarily in order to reduce the processing times of the following jobs (\cite{kubzin2006, rustogi2012}). However, just lately, position dependent maintenance operations have been introduced by Drozdowski, Jaehn and Paszkowski \cite{drozdowski2016}. 
Amongst others, they motivate position dependent maintenance activities with wear and tear of jet engines or aircraft wheels, which is rather caused by the number of flights (because of the climb flight and thrust reversal for the engines) than by the length of the flight. So the problem considered here can be seen as the special case in which exactly one position dependent maintenance activity is necessary.

Secondly, our problem is a special case of scheduling with non-negative inventory constraints, as was introduced by Briskorn et al. \cite{briskorn2010}. Here, each job either delivers or removes items of a homogeneous good to or from an inventory. Jobs that remove items can only be processed if the required number of items are available, i.e.\ only if the inventory level does not become negative. This problem relates to ours, in which a job is fixed to position $k$, as follows. The job fixed on position $k$ can be considered as the only job removing items from the inventory, and $k-1$ jobs are required to deliver items before this job can be scheduled. If the parameter settings of the fixed job are chosen such that this job is to be scheduled as early as possible, it is forced to be on position $k$. Analogously, the fixed job can be modeled as the only one delivering to the inventory so that it must be scheduled the latest on position $k$. Parameter settings then need to ensure that it is not scheduled earlier.

In this paper we continue the work of \cite{briskorn2010} on problem $1|inv|\sum w_jC_j$.
We consider one machine with the above mentioned inventory constraint with the objective of minimizing total weighted completion time. 
Briskorn et al. \cite{briskorn2010} show that this problem is strongly NP-hard in the general case and they propose various special cases, which are easily solvable and some which are still open. They especially differ between the sets of jobs that deliver to the inventory and that remove goods from the inventory. As mentioned before, we consider a special case in which one of the two sets only consists of one job. For this problem setting, we propose a fully polynomial time approximation scheme (FPTAS). 

Several special cases and generalizations of problem $1|inv|\sum w_jC_j$ have been analyzed in the literature. 
Briskorn and Pesch \cite{briskorn2013b} consider the generalization with a maximum inventory level. They show that even finding a feasible solution is NP-hard and they propose heuristics. Another generalization is analyzed by Kononov and Lin \cite{kononov2010}. Here, each job consumes some items at the beginning of its processing time and it adds to the inventory a number of items at its completion time. They show NP-hardness of several special cases and present some approximation algorithms for further special cases. 
Morsy and Pesch \cite{morsy2015} consider a special case in which all jobs delivering to the inventory must be equal (concerning processing time, weight, and inventory modification) and the remaining jobs must also share some characteristics. For this setting, a 2-approximation is presented. Optimality criteria and an exact branch-and-bound algorithm for the standard problem $1|inv|\sum w_jC_j$ are proposed by Briskorn et al. \cite{briskorn2013a}.

There are some problems discussed in the literature, which are closely related to $1|inv|\sum w_jC_j$. First of all, Briskorn and Leung \cite{briskorn2013c} consider the problem with maximum lateness objective function. They propose some optimality criteria, lower bounds, and heuristics, which are then used in a branch-and-bound framework. There are various papers (\cite{gyorgyi2014, gyorgyi2015a, gyorgyi2015b, kis2015}) that analyze a ``non-renewable resource constraint'', which means that each job removes goods from the inventory, but the inventory is filled automatically at predetermined points in time. So contrary to the above mentioned inventory constraint, which was exclusively based on the job sequence, here the constraint is partially time based. The papers on this problem mostly focus on minimizing makespan. Only Kis \cite{kis2015} considers the same objective function and presents a strong NP-hardness proof and an FPTAS for a special case.   

We formulate the problem with the constraint that a fixed amount of jobs must be finished when the special job (refer to as pivot job) starts. 
We present an FPTAS in this paper. 
First we propose a dynamic programming whose running time depends on job processing times. 
To obtain an FPTAS, we use rounding technique: we round the job processing times so that they are polynomial in $n$ and $1/\epsilon$, then we obtain the optimal schedule for the rounded jobs via dynamic programming and apply that schedule to original jobs.
However, the rounding approach does not guarantee $(1+\epsilon)$ - approximation.
To make it work, we further discover an important property when the rounding technique fails: the job with the largest weight can not be scheduled after (or the same as) the job with the largest processing time. The reason behind is that when this property breaks, the objective value of the optimal schedule is large enough, which makes the dynamic programming solution good enough.
With this property, on one hand we apply the rounding technique and on the other hand we put these two special jobs before or after the pivot job accordingly and solve the subproblem.

The remainder of the paper is organized as follows. The problem formulation is given in Section~\ref{sec:form}. In Section~\ref{sec:dp} we propose two dynamic programs to solve the problem, with running time polynomial in job processing times and job weights, respectively. 
Then we use the dynamic programming to design an FPTAS in Section~\ref{sec:fptas}. 
In Section~\ref{sec:trimspace}, we present another FPTAS as a comparison.
In Section~\ref{sec:conclusion} we conclude our work.
\section{Formulation} \label{sec:form}
The instance of the problem is a set of $n$ jobs $J = \{1,2,...,n\}$, a specified job $c \in J$ and an integer $k \in [1, n]$.
Each job $j \in J$ is defined by its weight $w_j$ and its processing time  $s_j$ (or sometimes refereed to as workload, size).
A schedule $\sigma$ over instance $J$ is an order of jobs, we write $i \preceq_{\sigma} j$ (resp. $i \succeq_{\sigma} j$) meaning that job $i$ precedes (resp. succeeds) job $j$ or jobs $i,j$ are the same in schedule $\sigma$.
The completion time of a job in a feasible schedule is the time when the job finishes.
Assume that the machine is never idle unless there is no more job to be processed, we define $C_j^{\sigma} = \sum_{i \preceq_{\sigma} j}~s_i$ as the \emph{completion time} for each job $j \in J$ in schedule $\sigma$.
Also, we denote $\delta_j = \frac{w_j}{s_j}$ as the {\em density} of job $j$.
The objective is to minimize total weighted completion time on a single machine such that there are exactly $k - 1$ jobs scheduled before job $c$, i.e. 
$
\min_{\sigma} \sum_{j \in J} w_j C_j^{\sigma} ~~~\text{s.t.}~~~  k = | \{~j ~|~ j \preceq_{\sigma} c, ~j \in J~\} |
$
where $k$ is part of the input.

\begin{theorem}
In the optimal solution, jobs that are scheduled before (or after) job $c$ must follow \emph{Smith's order}.
\end{theorem}
\noindent
In classical \emph{Smith's Rule} \cite{smith1956various} (or Smith's order), jobs are executed in non-increasing order of its density $\delta_j$. 
Smith's Rule has been proven to be optimal when there is no position constraint of job $c$. However, Smith'r Rule does not work in this problem as in the special case where $w_c$ approaches to infinity and the jobs that are placed before the pivot job in the optimal solution should have the smallest processing times.
In the sequel, we assume job $c$ is indexed as $n$, and the remaining jobs $J \setminus \{c\}$ are sorted in Smith's order, i.e. 
$\delta_1 \ge \delta_2 \ge ... \ge \delta_{n-1}$.

\section{Dynamic Programming with Side Constraints}
\label{sec:dp}
In this section, we propose pseudo-polynomial dynamic programs to solve this problem.
Given integer $k$, job $c \in J$ and sets of jobs $J$, $H \subset J \setminus \{c\}, ~ B \subset J \setminus \{c\}$ such that $H \cap B = \emptyset$,
we aim to find the optimal schedule such that
\textbf{(i.)} jobs from $H$ are scheduled  before job $c$
\textbf{(ii.)} jobs from $B$ are scheduled after job $c$
\textbf{(iii.)} there are exactly $k-1$ jobs scheduled before job $c$.

We say that job $j$ is {\em assigned} when the order of job $j$ and job $c$ is determined and  {\em unassigned} otherwise.
Therefore, jobs $H \cup B$ are {\em assigned} and let $U = J \setminus (B \cup H \cup \{c\})$ be the {\em unassigned} jobs.
Let $\hat{s_j} = \ce{\lambda \cdot s_j}$ be the rounded job processing time of job $j \in J$ with a given parameter $\lambda$. We would see later that $\lambda$ is polynomial in $n$, $1/\epsilon$, and linear in $1/s_{max}$ where $s_{max}$ is the maximum job processing time among all unassigned jobs. In other words, we make sure that for each unassigned job $j$, $\hat{s_j}$ is polynomial in $n$ and $1/\epsilon$. Similarly, we could also round job weights with a different $\lambda$ , i.e. $\hat{w_j} = \ce{\lambda \cdot w_j},~ \forall j \in J$.
In the following part of this section, we give two dynamic programs to solve the rounded jobs based on job processing time in Section~\ref{dp-size} and based on job weight in Section~\ref{dp-weight}, and denote $f_{s}(k,c,J,H,B,\lambda)$ and $f_{w}(k,c,J,H,B,\lambda)$ as the optimal schedule (the order of jobs) returned by the corresponding dynamic programming for the rounded jobs respectively.

\subsection{Based On Job Processing Time} \label{dp-size}
We propose a dynamic programming with pseudo-polynomial running time. That is, we assume for each job that its processing time has already been rounded into integers, and the running time of the dynamic program is polynomial in the number of jobs and the maximum job processing time.

Given a set of jobs $J^{*} \subseteq J$, we denote $[J^{*}] = \sum_{ j \in J^{*}} \hat{s_j}$ as the total processing time of jobs $J^{*}$.
For a feasible schedule $\gamma$, let $U(\gamma) = \{i~|~i\in U,i \prec_{\gamma} c\}$ be the subset of jobs in $U$ which are scheduled before job $c$ in schedule $\gamma$.

Let $n = |J|$ and $\hat{S} = [U]$.
A {\em partial schedule} of jobs $J' \subseteq J$ assigns to each of these jobs $J'$ a valid completion time, making sure that each job could be finished within that valid completion time (i.e. jobs do not overlap). 
First, we try every possibility of the completion time of job $c$, i.e. we aim to find the optimal schedule $\sigma$ such that
$[ U(\sigma) ] = L$ where we test every possibility of $L$ from $\{0, 1, \ldots, \hat{S}\}$.
Hence, we denote $C_c(L) = L + \hat{s_c} + [H]$ as the completion time of job $c$ when $L$ is fixed.
Afterwards, we consider jobs $J' = \{1,...,j\}$ and focus on two parameters $e,E$ in the optimal schedule $\sigma$ where $|J' \cap U(\sigma)| = e $ and $[J' \cap U(\sigma)] = E$.
Finally, we test every possibility of $e,E$ in the dynamic programming.
\begin{definition}
Let $dp(e,E,j)$ be the total weighted completion time of jobs $J' = \{1,...,j\}$ in an optimal partial schedule such that 
there are $e$ jobs from $J' \cap U$ which are scheduled before job $c$ with total processing time $E$, 
where $e \in \{0, \ldots, k-1\}$, 
$j \in \{0, \ldots, n-1\}$, 
$E \in \{0, 1, \ldots, L\}$. 
$dp(e,E,j)$ is taken to be infinity if no such partial schedule exists.
\end{definition}

\noindent
To find the optimal schedule of jobs $J^{'}$, we fix the schedule of job $j$ and then solve the subproblem.
We show that the completion time of job $j$ could be calculated once job $j$ is determined to be scheduled before or after job $c$. 
\begin{lemma}\label{lmm-dp}
For $j = 0$, we have $dp(e,E,j) = 0$ if $e = 0, E = 0$ and $dp(e,E,j) = \infty$ otherwise.
For $j > 0$, we have $dp(e,E,j) = $
\[ \min
\left\{
\begin{array}{ll}
dp(e,E,j-1) + 
w_j \left(E +  [H \cap J^{'}] \right),   &~\mbox{if}~ j \in H \\

dp(e-1,E-\hat{s_j},j-1) + 
w_j \left(E +  [H \cap J^{'}] \right), & ~\mbox{if}~ j \in U, e > 0, E \geq \hat{s_j} \\

dp(e,E,j-1) + w_j 
\left( C_c(L) + ([U \cap J^{'}] - E) +  [B \cap J^{'}]\right),
& ~\mbox{if}~ j \in B \cup U
\end{array}
\right.
\]

$f_{s}(k,c,J,H,B,\lambda) = \min_{L \in [0,\hat{S}]} ~ w_c \cdot C_c(L) + dp(k - 1 - |H|,L,n-1)$

\end{lemma}

\begin{proof}
Let $\sigma$ be an optimal schedule.
Without loss of generality, we assume that $ [U(\sigma)] = L$ because we try every possibility of value $L$ from $\{0, 1, \ldots, \hat{S}\}$.
Then, we will prove that Lemma~\ref{lmm-dp} gives the 
correct optimal solution.

For the base case, we have $j = 0$, i.e. $J'=\emptyset$.  The lemma is correct because no feasible schedule exists when $e \not = 0$ or $E \not = 0$.
For the case $j>0$, we prove the lemma by claiming that we have tried every possibility for the schedule of job $j$. More specifically, we show that when job $j$ is scheduled before (or after) job $c$, the completion time of job $j$ could be computed directly. Therefore, we only need to try two possibilities (before or after job $c$) for the schedule of job $j$.

In the following we show that the completion time of job $j$ is correct.
Without loss of generality, we assume 
$|J' \cap U(\sigma)| = e $ and $[J' \cap U(\sigma)] = E$ because we try every possibility of parameter $e$ and $E$.
Given $j \in J$, we denote  $j^{\preceq}$ as the set of jobs that are scheduled no later than $j$ (inclusive) in the final optimal schedule.

If job $j$ is scheduled before job $c$ in the optimal schedule, i.e. $j \in c^{\preceq}$, then by Smith's Rule, we have 
$(J \setminus J^{'}) \cap j^{\preceq} = \emptyset$, i.e. jobs from $J \setminus J^{'}$ will not be scheduled before job $j$ in the optimal schedule because $j$ is scheduled before job $c$.
Therefore, $j^{\preceq} = (J^{'} \cap U(\sigma)) \cup (J^{'} \cap H)$, which implies that the completion time of job $j$ is $C_j = E + [J^{'} \cap H]$.
It corresponds to the first (resp. second) case of the equation, if $j \in H$ (resp. $j \in U$).

Otherwise, job $j$ is scheduled after job $c$ in the optimal schedule. 
By Smith's Rule, we have $J^{'} \subseteq j^{\preceq}$, i.e. jobs from $J^{'}$ must be scheduled before job $j$ in the optimal schedule because job $j$ is scheduled after job $c$. 
Moreover, we have $(J \setminus J^{'}) \cap (j^{\preceq} \setminus c^{\preceq}) = \emptyset$, i.e. jobs from $J \setminus J^{'}$ will not be scheduled between $c$ and $j$.
Therefore $j^{\preceq} \setminus c^{\preceq} = (J^{'} \cap (U \setminus U(\sigma)) ) \cup (J^{'} \cap B)$.
As a result, the completion time of job $j$ is 
$C_j = C_c(L) + [J^{'} \cap U] - E +  [J^{'} \cap B]$.
It corresponds to the third case of the equation, if $j \in B$ or $j \in U$.
\end{proof}

\noindent
\textbf{Time Complexity:} 
Note that the values $[H \cap J^{'}]$, $[B \cap J^{'}]$ and $[U \cap J^{'}]$ could be precomputed, and they will not change the overall running time. In other words, the running time depends on the unassigned jobs.
The overall time complexity is $O(n^4 \hat{s}_{max}^2)$ where $\hat{s}_{max} = \max_{j \in U} \hat{s_j}$. 
Indeed, the dynamic programming has a table size $O(n^2 \hat{S})$, the computation of each value $dp(e,E,j)$ takes $O(1)$ operations, and the dynamic programming needs $O(\hat{S})$ time for $L \in \{0,1,...,\hat{S}\}$, thus in total the time complexity is
$O(n^2 \hat{S}^2)  = O(n^4 \hat{s}_{max}^2)$.
It is important that $\hat{s}_{max}$ only depends on the unassigned jobs $U$.

\subsection{Based On Job Weight}\label{dp-weight}
In this section, the unassigned jobs $U$ are required to have integer weight, as the running time of the dynamic programming depends on the weights of the unassigned jobs.
We assume for each job that its weight is already rounded to integer.
As this problem is highly symmetrical, we show that the dynamic programming in Section \ref{dp-size} could be applied by Theorem~\ref{thm-weight}.
For each job $j \in J$, we create a corresponding job $j^{*}$ with processing time $w_j$, weight $s_j$, and we obtain a new instance $J^{*}$, i.e. $\forall j^{*} \in J^{*}, ~w_{j^{*}} = s_j,~s_{j^{*}} = w_j$.

\begin{theorem} \label{thm-weight}
The reverse order of $f_{s}(n + 1 - k,c^{*},J^{*},B^{*},H^{*}, \lambda)$ is $f_{w}(k,c,J,H,B, \lambda)$.
\end{theorem}

\begin{proof}
We denote $obj(I,k,\sigma)$ as the objective value of schedule $\sigma$ for the jobs $I$ with parameter $k$ (position constraint parameter).
Let $\pi$ be any feasible schedule for jobs $J$ with parameter $k$, and let $\pi^{'}$ be the reverse of $\pi$, i.e. $i \preceq_{\pi} j$ if and only if $j \preceq_{\pi^{'}} i$. We prove that
\[
obj(J,k,\pi) = obj(J^{*},n+1-k,\pi^{'})
\]

Firstly, in schedule $\pi$ there are $k-1$ jobs which are scheduled before job $c$ since $\pi$ is feasible for $J$ with parameter $k$.
Therefore, in schedule $\pi^{'}$ for $J^{*}$ there are $n-k$ jobs which are scheduled before job $c^{*}$ by definition of $\pi^{'}$. 
Moreover, in schedule $\pi$ jobs $H$ are scheduled before job $c$, then in schedule $\pi^{'}$ jobs $H^{*}$ are scheduled after job $c^{*}$. Similar analysis can be used for jobs $B$.
Consequently, schedule $\pi^{'}$ is a feasible schedule for $J^{*}$ with parameter $n+1-k$.
Then, we prove the theorem by the equation:
$
obj(J,k,\pi) = \sum_{j \in J} w_j \sum_{i \preceq_{\pi} j} s_i
=  \sum_{i \in J} s_i \sum_{i \preceq_{\pi} j} w_j 
=  \sum_{i \in J} w_{i^{*}} \sum_{i \preceq_{\pi} j} s_{j^{*}} 
=  \sum_{i \in J} w_{i^{*}} \sum_{j \preceq_{\pi^{'}} i} s_{j^{*}} 
= obj(J^{*},n + 1 - k,\pi^{'})
$.
In the first equality, we formulate the objective of schedule $\pi$ for $J$.
In the second equality, we reorganize the summation. 
In the third equality, for each $j \in J$ we substitute $w_j$ by $s_{j^{*}}$ and $s_j$ by $w_{j^{*}}$ as they have equal value. 
In the fourth equality, we replace
$\pi$ by $\pi^{'}$.

Consequently, the reverse order of the optimal solution for jobs $J^{*}$ with parameter $n+1-k$ is optimal for jobs $J$ with parameter $k$.
\end{proof}

\section{ Fully Polynomial-Time Approximation Scheme(FPTAS)} \label{sec:fptas}
In this section, we present the FPTAS algorithm.
Recall that the dynamic programming in previous section gives the optimal solution while the running time depends on job processing times (or job weights). 
The straightforward idea is to round the job processing times such that they are polynomial in $n$ and $1/\epsilon$ and then solve the rounded jobs via dynamic programming.
However, this technique does not guarantee $(1+\epsilon)$ - approximation where we will show an example.
Later, we extract information from this failure and design an FPTAS.

\paragraph*{Rounding Technique}
For each job $j \in J$, we round job processing time with parameter $\lambda$, i.e. $\hat{s}_j = \ce{\lambda \cdot s_j}$ with $\lambda = \frac{h(n,\frac{1}{\epsilon})}{ s_{max} }$ where $s_{max}$ is the maximum processing time of all jobs $J$ and $h(n,\frac{1}{\epsilon})$ is a function which is polynomial in $n$ and $1/\epsilon$.
We obtain the optimal schedule (denote by $\sigma$) for the rounded jobs via dynamic programming and analyze the performance of schedule $\sigma$ for jobs $J$.
Let $\pi$ be the optimal schedule for jobs $J$. The objective of $\sigma$ could be bounded:
\begin{equation}\label{eq:dp-round}
\begin{array}{l}
\sum_{j \in J} w_j \sum_{i \preceq_{\sigma} j} s_i 
 \le \sum_{j \in J} w_j \sum_{i \preceq_{\sigma} j} \hat{s}_i/\lambda 
 \le \sum_{j \in J} w_j \sum_{i \preceq_{\pi} j}\hat{s}_i/\lambda \\
 \le \sum_{j \in J} w_j \sum_{i \preceq_{\pi} j} (\lambda s_i + 1)/\lambda
 =  \opt(J) + \sum_{j \in J} w_j \sum_{i \preceq_{\pi} j} 1/\lambda\\
 \le  \opt(J) + (n/\lambda) \cdot \sum_{j \in J} w_j\end{array}
\end{equation}
\noindent
where in the first and third inequality we use $\lambda s_j \le \hat{s}_j \le \lambda s_j + 1  $, and in the second inequality we apply the fact that $\sigma$ is optimal for the rounded jobs.
Similarly, when we round job weights with a different parameter $\lambda$, i.e. $\hat{w}_j = \ce{\lambda \cdot w_j},\forall j \in J$, we would have
\begin{equation}\label{eq:dp-roundw}
\begin{array}{l}
\sum_{j \in J} w_j \sum_{i \preceq_{\sigma} j} s_i 
\le \opt(J) + (n/\lambda) \cdot \sum_{i \in J} s_i
\end{array}
\end{equation}

The error $(n/\lambda) \cdot \sum_{j \in J} w_j$ in Equation~\eqref{eq:dp-round} may not be bounded by $\epsilon \cdot \opt(J)$, where one would see from the following example.
In the example, we have $3$ jobs where $s_1 = 1,s_2 = 2$, $s_3 \gg 2 h(n,\frac{1}{\epsilon})$ and $w_1 = s_3 , w_2 = w_1+1, w_3 = 1$.
After rounding, $\hat{s}_1 = \hat{s}_2 = 1$ as $s_{max} = s_3$, therefore the optimal schedule for the rounded jobs will be $\sigma = (2 \prec 1 \prec 3)$, while the optimal schedule for original jobs is  $\pi = (1 \prec 2 \prec 3)$.
Therefore, the approximation ratio is $\frac{w_2*2+w_1*(1+2)+1*(1+2+s_3)}{w_1*1+w_2*(1+2)+1*(1+2+s_3)} \ge 17/16$, which is a constant.

Note that the error in Equation~\eqref{eq:dp-round} is $(n/\lambda) \cdot \sum_{j \in J} w_j = 
\frac{ n s_{max} \cdot \sum_{j \in J} w_j }{h(n,\frac{1}{\epsilon})}
$.
This error may not be bounded by $\epsilon \cdot \opt(J)$ if the objective value of optimal solution is small, comparing to the product of maximum job processing time and maximum job weight. 
Therefore, we focus on two such special jobs, job $u = \argmax_{j \in J} \nolimits ~\{s_j\}$ the job of largest processing time, and job $v = \argmax_{j \in J} \nolimits ~\{w_j\}$ the job of largest weight. 
Note that $s_{max} = s_u$ and that if job $v$ is scheduled after job $u$ or $u = v$ in the optimal solution, i.e. $v \succeq_{\pi} u$, we will have
$\opt(J) \ge w_v s_u$, then the error in Equation~\eqref{eq:dp-round} could be bounded when we 
take $h(n,\frac{1}{\epsilon}) = n^2 \cdot \frac{1}{\epsilon}$:
\[
(n/\lambda) \cdot \sum_{j \in J} w_j 
\le n^2 w_v / \lambda 
=  \frac{n^2 w_v \cdot s_u }{h(n,\frac{1}{\epsilon}) }
\le \frac{ n^2 }{ h(n,\frac{1}{\epsilon}) } \cdot \opt(J)  \le \epsilon \cdot \opt(J)
\]
Therefore, when the rounding technique fails, we would have $v \prec_{\pi} u$, i.e. the job of largest weight must be scheduled before the job of largest processing time. 
This property from the failure of rounding technique plays an important role in designing the FPTAS algorithm.

\begin{algorithm}[H]
\caption{FPTAS Algorithm $F\langle c,k,H,B,U\rangle$}\label{alg:ptas}

\begin{algorithmic}[1]

\INPUT Consisting of specified job $c$,
specified value $k$, 
set of jobs $H$ (resp. $B$) assigned to be scheduled before (resp. after) $c$, 
and set of unassigned jobs $U$. Here, $\{c\}, H, B, U$ are pairwise disjoint, and 
$H = \emptyset$ if and only if $B = \emptyset$.

\OUTPUT A sequence $\chi$ of jobs $U \cup H \cup B \cup \{c\}$ with exactly $k-1$ jobs scheduled before job $c$, or $\emptyset$ if no such sequence exists.

\State $J_r \gets U \cup H \cup B \cup \{c\}$, $u \gets \argmax_{j \in U} \nolimits ~\{s_j\}$, $v \gets \argmax_{j \in U} \nolimits ~\{w_j\}$ 
\State $\chi \gets $ an arbitrary feasible schedule of $J_r$
\Comment{best from following cases}

\State call $\textsc{CheckFeasibility}()$
\State call $\textsc{FixJob}()$
\State call $\textsc{RepeatSize}()$
\Comment{If $v \succ c$ in the optimal schedule}
\State call $\textsc{RepeatWeight}()$
\Comment{If $u \prec c$ in the optimal schedule}

\If{ $u \neq v $}
\Comment{If $v \prec c \prec u$ in the optimal schedule}
\State $\sigma \gets F\langle c,k,H \cup \{v \}, B \cup \{ u \},U \setminus \{u,v\}\rangle$, Update $\chi \gets best \{\chi,\sigma\}$
\label{alg-recursive}
\EndIf

\State \textbf{Return} $\chi$

\algstore{alg:break}
\end{algorithmic}
\end{algorithm}

\begin{algorithm}

\begin{algorithmic} [1]
\algrestore{alg:break}

\Procedure{CheckFeasibility}{}
\Comment{Check Feasibility}
\If{$|H| > k-1$ or $|H| + |U| < k -1 $}
\State \textbf{Return}	$\emptyset$
\ElsIf{$  |H| = k - 1  $} \Comment{termination case}
\State $ B \gets B \cup U$
\ElsIf{$|H| + |U| = k -1 $} \Comment{termination case}
\State $ H \gets H \cup U$
\EndIf

\State Sort jobs $H$ and $B$ by Smith's order respectively

\State $\chi \gets (H,c,B)$, \textbf{Return} $\chi$.

\EndProcedure
\;
\\\hrulefill

\Procedure{FixJob}{}
\Comment{Fix One Job}

\State $S\gets \sum_{j \in J_r} s_j$, $W\gets \sum_{j \in J_r} w_j$, $S^{'}\gets \sum_{j \in U} s_j$, $W^{'}\gets \sum_{j \in U} w_j$

\If {$S^{'} \le \epsilon S/n$ \textbf{and} $B = \emptyset$}
\Comment{termination case}
\State $H^{'} \gets$ the first $k-1$ jobs from $U$ by non-increasing order of job weight.
\State \textbf{Return} $F\langle c,k,H^{'},U\setminus H^{'}, \emptyset\rangle$
\label{alg-place-weight}

\ElsIf{$W^{'} \le \epsilon W/n$ \textbf{and} $H = \emptyset$}
\Comment{termination case}
\State $H^{'} \gets$ the first $k-1$ jobs from $U$ by non-decreasing order of job processing time.
\State \textbf{Return} $F\langle c,k,H^{'},U\setminus H^{'}, \emptyset\rangle$
\label{alg-place-size}

\ElsIf {$S^{'} \le \epsilon S/n$ \textbf{and} $B \neq \emptyset$}
\Comment{place job $i$ as the last job, i.e. $i \succeq J_r$}
\State $i \gets \argmin_{j \in B} \{\delta_j\}$, $\chi' \gets F\langle c,k,\emptyset, \emptyset,
U \cup B \cup H \setminus \{i\} \rangle$. 
\State \textbf{Return} $(\chi',i)$ 
\label{alg-place-last}

\ElsIf {$W^{'} \le \epsilon W/n$ \textbf{and} $H \neq \emptyset$}
\Comment{place job $i$ as the first job, i.e. $i \preceq J_r$}
\State $i \gets \argmax_{j \in H} \{\delta_j\}$, 
$\chi' \gets F\langle c,k-1,\emptyset, \emptyset,
U \cup B \cup H \setminus \{i\} \rangle$.
\State \textbf{Return} $(i,\chi')$
\label{alg-place-first}
\EndIf

\EndProcedure 
\;
\\\hrulefill

\Procedure{RepeatSize}{}
\Comment{round job processing time}
\State copy $U,H,B$
\Comment{make copy of jobs}
\While{$|U| + |H| \ge k$}
\State $p \gets \argmax_{j \in U} \{s_j\}$, $\lambda \gets \frac{n^3}{\epsilon^2 s_{p}}$, 
\State $\sigma \gets f_s (k,c,J_r,H,B,\lambda)$, Update $\chi \gets best \{\chi,\sigma\}$.
\State $B \gets B \cup \{ p \}$, $U \gets U \setminus\{ p\}$.
\EndWhile 
\State $\sigma \gets F\langle c,k,H \cup U,B,\emptyset \rangle$, 
Update $\chi \gets best \{\chi,\sigma\}$
\EndProcedure
\;
\\\hrulefill

\Procedure{RepeatWeight}{}
\Comment{round job weight}
\State copy $U,H,B$
\Comment{make copy of jobs}
\While{$ |H| < k - 1$}
\State $q \gets \argmax_{j \in U} \{w_j\}$, 
$\lambda \gets \frac{n^3}{\epsilon^2 w_{q}}$, 
\State 
$\sigma \gets f_w (k,c,J_r,H,B,\lambda)$,
Update $\chi \gets best \{\chi,\sigma\}$.
\State $H \gets H \cup \{ q \}$, $U \gets U \setminus\{ q\}$.
\EndWhile 
\State $\sigma \gets F\langle c,k,H,B \cup U,\emptyset \rangle$, 
Update $\chi \gets best \{\chi,\sigma\}$
\EndProcedure
\end{algorithmic}
\end{algorithm}

\paragraph*{FPTAS Algorithm}
From the above analysis, the rounding technique could possibly fail to return a good solution, which we never know. 
In case that the failure happens, we would assign some jobs based on the property from such failure that the objective value of the optimal schedule is small (i.e. the job of largest weight must be scheduled before the job of largest processing time).
Afterwards we run the rounding technique again, and still a good solution may not be returned.
Indeed, we could recursively assign jobs and  apply the rounding technique. However, as more and more jobs are assigned, the unassigned jobs will have small weight and processing time, which will not reflect the objective value of the optimal schedule. 
In other words, the above property will not hold any more.
Instead, we would fix the position of one job when such case happens, i.e. the job weight or job processing time of the unassigned job is small enough.

The FPTAS algorithm has many rounds. 
In each round, we aim to fix the position of one job.
More precisely, we make this job as the first job (or last job), then we take the remaining jobs as a new instance (update constraint parameter $k$ accordingly) and start over.
We guarantee that the performance of the solution lose by a factor of $(1+\epsilon/n)$ each time when we fix one job.
Let $J_r$ be the remaining jobs in current round, i.e. jobs $J \setminus J_{r}$ are already fixed. 
We would take $J_r$ as an instance, and let $\pi$ be the optimal schedule of jobs $J_r$.
In order to find and fix one job from $J_r$, the algorithm will go into many iterations and assign jobs into sets $H \subset J_r,B \subset J_r$ such that either $H = B = \emptyset$ or $H \neq \emptyset,~B \neq \emptyset$, where jobs $H$ (resp. $B$) are determined to be scheduled before (resp. after) job $c$. 
Let $U = J_r \setminus (H \cup B \cup \{c\})$ be the unassigned jobs.
Let $S = \sum_{j \in J_r} s_j$, $W = \sum_{j \in J_r} w_j$ and $S^{'} = \sum_{j \in U} s_j$, $W^{'} = \sum_{j \in U} w_j$.
The algorithm handles the problem separately according to the following inequalities.
\begin{align}
S^{'} & \le \epsilon S / n \label{ineq-S}\\
W^{'} & \le \epsilon W / n \label{ineq-W}
\end{align}
\begin{lemma}\label{lmm:fixjob}
Assume that the optimal schedule assign jobs $H$ (resp. jobs $B$) before (resp. after) job $c$, if inequality \eqref{ineq-S} or \eqref{ineq-W} holds, we are able to either

\noindent \textbf{i)} obtain a feasible schedule with $(1+\epsilon)$-approximation, or

\noindent \textbf{ii)} fix one job from $J_r$ as the first job or last job by losing at most a factor of $(1+\epsilon/n)$ comparing to the optimal schedule of $J_r$.

\end{lemma}
\begin{proof}
In the following cases, we claim that i) could be achieved if Case 1) or 2) happens and ii) could be achieved if Case 3) or 4) happens (refer to Algorithm \ref{alg:ptas} procedure \textsc{FixJob}).

\textbf{Case 1)} If $S^{'}  \le \epsilon S / n$ and $B = H = \emptyset$, we assign jobs $H^{'}$ before job $c$, jobs $U \setminus H^{'}$ after job $c$ and terminate the algorithm, where $H{'}$ are the first $k-1$ jobs from $U$ by non-increasing order of job weight. 
Let $\chi$ be the corresponding schedule.
In the optimal schedule $\pi$, let $L \subset J_r$ (resp. $R \subset J_r$) be the set of jobs scheduled before (resp. after) job $c$. 
In schedule $\chi$, we use the corresponding notation $\tilde{L},\tilde{R}$.
Schedule $\psi = (L\cap \tilde{L},R\cap \tilde{L},\{c\}, \tilde{R})$ is obtained based on schedule $\chi$ by advancing and rearranging jobs $L\cap \tilde{L}$ as the order in the optimal schedule, 
hence the completion time of jobs $L\cap \tilde{L}$ in schedule $\psi$ is at most that in the optimal schedule, i.e. $C^{\psi}_j \le C_j^{\pi}, ~ \forall j \in L \cap \tilde{L}$.
Since jobs $\tilde{L}$ (resp. $\tilde{R}$) in schedule $\chi$ are ordered by Smith Rule, the objective value of $\chi$ is at most that of $\psi$.
Note that $S = S' + s_c$ as $B = H = \emptyset$, then we have $C^{\psi}_j \le S - s_c \le \epsilon S /n, ~ \forall j \in R \cap \tilde{L}$ because these jobs are scheduled before job $c$, and $C^{\psi}_j \le S, ~ \forall j \in \tilde{R}$.
Thus the objective value of $\psi$ is at most:
\[
\begin{array}{ll}
&\sum_{j \in L \cap \tilde{L}} \nolimits w_j C_j^{\pi} + \sum_{j \in R \cap \tilde{L}} \nolimits  w_j (\epsilon S /n) + \sum_{j \in \tilde{R} \cup \{c\}} \nolimits w_j S \\
 &\le
\sum_{j \in L \cap \tilde{L}} \nolimits w_j C_j^{\pi} + (1 + \epsilon /n ) S \sum_{j \in R \cup \{c\}} \nolimits w_j \\
 &\le \frac{1 + \epsilon /n}{1 - \epsilon /n} \sum_{j \in J_r} w_j C_j^{\pi}
\end{array}
\]
where in the first inequality we apply $\sum_{j \in \tilde{R}} w_j \le \sum_{j \in R} w_j $ as jobs $\tilde{L}$ are selected by job weight from jobs $U$,
and in the second inequality we apply $C_j^{\pi} \geq s_c \ge S(1-\epsilon/n), ~\forall j \in R \cup \{c\} $.
As $\frac{1 + \epsilon/n}{1- \epsilon/n} \leq 1 + \epsilon$
for $n \geq 3$ (one would enumerate all possible solutions for $n \leq 2$).
The claim follows.

\textbf{Case 2)} If $W^{'}  \le \epsilon W / n$ and $B = H = \emptyset$, we assign jobs $H^{'}$ before job $c$, jobs $U \setminus H^{'}$ after job $c$ and terminate the algorithm, where $H{'}$ are the first $k-1$ jobs from $U$ by non-decreasing order of job processing time. A similar argument could be constructed as Case 1).

\textbf{Case 3)} If $S^{'}  \le \epsilon S / n$ and $B \not = \emptyset$, we place job 
$i = \argmin_{j \in B} \{\delta_j\}$ as the last job among $J_r$ and reduce to subproblem by taking the remaining jobs $J_r\setminus\{i\}$ as a new instance.
Let $\chi$ be the schedule transformed from $\pi$ by placing job $i$ after jobs $J_r \setminus \{i\}$.
Schedule $\chi$ is feasible because job $i$ must be scheduled after job $c$ in the optimal schedule as $i \in B$.
Hence, after transformation the completion time of any job of $J_r \setminus \{i\}$ in schedule $\chi$ is at most that in schedule $\pi$.
By assumption, in the optimal schedule $\pi$, job $i$ must be scheduled after all jobs in $(H \cup B \cup \{c\}) \setminus \{i\}$, we have $C_i^{\pi} \ge \sum_{j \in J_r \setminus U } s_j = S - S^{'}$ and $C_i^{\chi} = S$.
Therefore,
\[
\frac{C_i^{\chi}}{C_i^{\pi}} \le \frac{S}{S-S^{'}} \le 1 + \frac{\epsilon}{n-\epsilon}
\]

\textbf{Case 4)} If $W^{'}  \le \epsilon W / n$ and $H \not = \emptyset$, we place job 
$i = \argmax_{j \in H} \{\delta_j\}$ as the first job among $J_r$. 
A similar argument could be constructed as Case 3).
\end{proof}

\begin{lemma}
Assume the rounding technique fails to return $(1+\epsilon)$-approximation solution every time, then either inequality \eqref{ineq-S} or \eqref{ineq-W} will hold after at most $n$ iterations.
\end{lemma}
\begin{proof}
Initially, we have $H = B = \emptyset, U = J_r \setminus \{c\}$.
Suppose $S^{'}  > \epsilon S / n $ and $W^{'}  > \epsilon W / n$.
Let job $u = \argmax_{j \in U} \nolimits ~\{s_j\}$ be the job of largest processing time among unassigned jobs, and let $v = \argmax_{j \in U} \nolimits ~\{w_j\}$.
In each iteration, we first apply the rounding technique (round job processing time) with parameter $ \lambda = \frac{n^3}{\epsilon^2 s_{u}}$. 
Note that the time complexity of dynamic programming in Section~\ref{sec:dp} only depends on unassigned jobs.

If $v \succeq_{\pi} u$, we claim that the rounding technique will return $(1+\epsilon)$-approximation solution due to $W^{'}  > \epsilon W / n$, as the error in Equation~\eqref{eq:dp-round} could be bounded
\[
(n/\lambda) \cdot \sum_{j \in J_r} w_j
= \frac{\epsilon^2 s_u}{n^2} \cdot W
< \frac{\epsilon s_u}{n} \cdot W^{'}
\le \epsilon \cdot w_v s_u 
\le \epsilon \cdot \opt(J_r)
\]

Otherwise, $v \prec_{\pi} u$, we solve by three cases.

\textbf{Case 1)} If $c \prec_{\pi} v \prec_{\pi} u$. We assign job $u$ into set $B$ (as the optimal does) and remove job $u$ from $U$, i.e. $U^{'} = U \setminus \{u\}$. 
Then we apply the rounding technique (round job processing time) again with $ \lambda^{'} = \frac{n^3}{\epsilon^2 s_{u^{'}}}$ where $u^{'} = \argmax_{j \in U^{'}} \nolimits ~\{s_j\}$.
Note that job $v$ is still the job of largest weight among unassigned jobs $U^{'}$. We claim that we would have $v \prec_{\pi} u^{'}$ if the rounding technique fails again.
Similarly, if $v \succeq_{\pi} u^{'}$, we have 
\[
(n/\lambda^{'}) \cdot \sum_{j \in J_r} w_j
= \frac{\epsilon^2 s_{u^{'}}}{n^2} \cdot W
< \frac{\epsilon s_{u^{'}}}{n} \cdot W^{'}
\le \epsilon \cdot w_v s_{u^{'}} 
\le \epsilon \cdot \opt(J_r)
\]
That is, the rounding technique returns $(1+\epsilon)$-approximation solution and the claim follows.
Therefore, we have $c \prec_{\pi} v \prec_{\pi} u^{'}$ if the rounding technique fails again, which implies that the algorithm will stay on Case 1).
Hence, we continue assigning job $u^{'}$ into set $B$ until at some moment $v = u^{'}$ (refer to Algorithm \ref{alg:ptas} procedure \textsc{RepeatSize}).
Consequently, at least one rounding technique will succeed.

\textbf{Case 2)} If $v \prec_{\pi} u \prec_{\pi} c$. 
We apply the rounding technique of rounding job weights by taking $ \lambda = \frac{n^3}{\epsilon^2 w_{v}}$ (refer to Algorithm \ref{alg:ptas} procedure \textsc{RepeatWeight}).
The error in Equation~\eqref{eq:dp-roundw} could be bounded due to $S^{'}  > \epsilon S / n $:
\[
(n/\lambda) \cdot \sum_{j \in J_r} s_j
= \frac{\epsilon^2 w_v}{n^2} \cdot S
< \frac{\epsilon w_v}{n} \cdot S^{'}
\le \epsilon \cdot w_v s_u 
\]
A similar argument could be constructed to show that once Case 2) is triggered, the algorithm will stay on Case 2) and at least one rounding technique will succeed.

\textbf{Case 3)} If $v \prec_{\pi} c \prec_{\pi} u$. We assign 
job $v$ into set $H$ and job $u$ into set $B$, and continue to the next iteration (refer to Algorithm \ref{alg:ptas} line \ref{alg-recursive}).

\noindent
One would see that there will be at most $n$ iterations to assign all jobs, i.e. the procedure \textsc{RepeatSize} and \textsc{RepeatWeight} in Algorithm \ref{alg:ptas} will be called at most $n$ times.
\end{proof}

\begin{lemma}
Algorithm \ref{alg:ptas} is $(1 + 3 \epsilon)$ - approximation with time complexity $O(n^{13}/\epsilon^4 )$.
\end{lemma}

\begin{proof}
We first prove that the algorithm will give 
$(1 + 3 \epsilon)$ - approximation solution.
The solution returned by the algorithm comes from either a successful rounding technique, or the termination case in the procedure \textsc{FixJob} (Line~\ref{alg-place-weight} and \ref{alg-place-size}), and before that the algorithm may have already fixed some jobs $J \setminus J_r$ (Line~\ref{alg-place-last} and \ref{alg-place-first}).
By Lemma~\ref{lmm:fixjob}, the termination case will return $(1+\epsilon)$-approximation solution, also a successful rounding technique will return $(1+\epsilon)$-approximation solution, comparing to $\opt(J_r)$.
When fixing one job, we would lose a factor of $(1+\epsilon/n)$ by Lemma~\ref{lmm:fixjob}, comparing $\opt(J_r)$.
One may think about transforming the optimal solution of jobs $J$ into our solution by fixing one job each time, then each time we still lose a factor of $(1+\epsilon/n)$ comparing to $\opt(J)$.
Therefore, the total approximation of fixing jobs will increase exponentially, which is $(1+\epsilon/n)^n = 1 + \epsilon + o(\epsilon^2)$.
Finally, the overall approximation is 
$(1 + \epsilon + o(\epsilon^2)) (1 + \epsilon) < 1 + 3 \epsilon$.

We show the time complexity of the algorithm.
After rounding, the largest job processing time or job weight of unassigned jobs is at most $O(n^3/\epsilon^2)$ as we take $ \lambda = \frac{n^3}{\epsilon^2 s_{u}}$ when rounding job processing time or $ \lambda = \frac{n^3}{\epsilon^2 w_{v}}$ when rounding job weight.
Therefore, each dynamic programming has running time $O(n^{10}/\epsilon^4)$, 
In each iteration, the dynamic programming procedure \textsc{RepeatSize} and \textsc{RepeatWeight} is called once, each procedure executing dynamic programming at most $n$ times.
We need to try $O(n)$ iterations to fix one job (Line~\ref{alg-place-last} and \ref{alg-place-first}) or terminate the algorithm (Line~\ref{alg-place-weight} and \ref{alg-place-size})
We need to fix at most $n$ jobs. 
Finally, the time complexity is $O(n^{13}/\epsilon^4 ) $.
\end{proof}

\section{Different Approach}\label{sec:trimspace}
In this section, we show that a different FPTAS could be constructed based on the approach by Woeginger \cite{woeginger2000does}.
Woeginger proposed some conditions to identify whether a dynamic programming could be transformed into an FPTAS, using the method of trimming state space.

First, we present a different dynamic programming, and then show that the conditions are satisfied.
We assume that jobs have integer weights and integer processing times.
Recall that job $c$ is indexed as $n$, and the remaining jobs $J \setminus \{c\}$ are sorted by Smith's order.
We start from the schedule which only contains job $c$, and then add remaining jobs into the schedule one by one.
The dynamic programming algorithm works with vector sets $\vs_1,...,\vs_{n-1}$ where in phase $j$ ($1 \le j \le n-1$) job $j$ is considered and $\vs_{j}$ is computed from $\vs_{j-1}$.
A state vector $[i,X,Y,Z]$ in $\vs_{j}$ encodes a partial schedule without idle time for jobs $\{1,2,...,j\} \cup \{c\}$: $i$ is the number of jobs that are scheduled before job $c$, $X$ (resp. $Y$) is the total processing time (resp. total weights) of jobs before (resp. after and include) job $c$, and $Z$ is the objective value for the partial schedule.

{\bf Initialization.} Set $\vs_0 = \{[0,0,w_c,w_c s_c]\}$.

{\bf Phase j.} For every vector $[i,x,y,z]$ in $\vs_{j-1}$, put the vectors
$[i,x,y+w_j,z+w_j(s_c + \sum_{i=1}^j s_i)]$ and
$[i+1,x+s_j,y,z+w_j(x+s_j) + y s_j]$ into $\vs_j$.

{\bf Output.} return the vector $[i,x,y,z] \in \vs_{n-1}$ that minimizes the $z$ value such that $i=k-1$.

Note that in phase $j$, by Smith's rule job $j$ can only be scheduled at the end or right before job $c$. 
If job $j$ is scheduled at the end, we just append job $j$ into the schedule, then the objective value only increases by the weighted completion time of job $j$, i.e. $w_j(s_c + \sum_{i=1}^j s_i)$.
Otherwise job $j$ is scheduled right before job $c$. When we insert job $j$ right before job $c$, besides the weighted completion time of job $j$, the completion time of those jobs that are scheduled after job $j$ will increase by $s_j$, therefore the objective value will increase by $w_j(x+s_j) + y s_j$.
Since the coordinates of all vectors in all sets $\vs_j$ are integers, the cardinality of every vector set $\vs_{j}$ is bounded from above by $O(n W^2 S^2)$, therefore the dynamic programming algorithm has a pseudo-polynomial time complexity of $O(n^2 W^2 S^2)$ where $W = \sum_{j \in J} w_j$ and $S = \sum_{j \in J} s_j$.

We show that an FPTAS exists from the conditions.
For $j=1,...,n$ we define the input vector $X_j = [w_j,s_j]$.
Let $\mathcal{F} = \{F_1,F_2\}$ where
\[
\begin{array}{ll}
F_1(w_j,s_j,i,x,y,z) &= [i,x,y+w_j,z+w_j(s_c + \sum_{i=1}^j s_i)]\\
F_2(w_j,s_j,i,x,y,z) &= [i+1,x+s_j,y,z+w_j(x+s_j) + y s_j]
\end{array}
\]
Let degree-vector $D=[0,1,1,1]$ and objective function $G(i,x,y,z) = z$ if $i=k-1$ and $\infty$ otherwise.
Then according to the approach by Woeginger \cite{woeginger2000does}, this dynamic programming is ex-benevolent and an FPTAS could be constructed.
Especially, the time complexity of the FPTAS is 
$
O(n^2 \log^2_{\Delta}{W} \log^2_{\Delta}{S})
= O(n^6 \log^2{W} \log^2{S}/ \epsilon^4)
$
where $\Delta = 1+\epsilon/n$.

\section{Conclusion and Discussion}\label{sec:conclusion}

We study the classical scheduling problem where one specific job must be scheduled at a specified position.
We give pseudo-polynomial time dynamic programs to solve this problem, which are polynomial in job processing time and job weight, respectively.
Moreover, we design a fully polynomial-time approximation scheme (FPTAS) of $(1+3 \epsilon)$-approximation with running time $O(n^{13}/\epsilon^4 ) $.
Our first method of using rounding technique based on pseudo-polynomial dynamic programs and fixing items one at a time in case the rounding fails may have further applications.
For the new approach of FPTAS in Section \ref{sec:trimspace}, the running time depends on the job weight ($\log{W} $) and job processing time ($\log{S} $), while our first FPTAS algorithm does not.

It remains open whether this problem is NP-hard or not.
It will also be interesting to study the multiple position constraints, i.e. two or more jobs have fixed positions. In this setting, our dynamic programming algorithm could be easily generalized and the rounding technique will not change too much while it is difficult to design the FPTAS algorithm to avoid exponential number of recursions.

\section*{Acknowledgment}
The work described in this paper was supported by a grant from Research Grants Council of the Hong Kong Special Administrative Region, China (Project No. CityU 11268616).

\bibliography{isaac2017}

\end{document}